\documentclass{vgtc}                          

\let\ifpdf\relax
\newtheorem{theorem}{Theorem}[section]

\newenvironment{proof}[1][Proof]{\begin{trivlist}
\item[\hskip \labelsep {\bfseries #1}]}{\end{trivlist}}

\usepackage{mathptmx}
\usepackage{graphicx}
\usepackage{times}
\usepackage{url}
\usepackage{hyperref}
\usepackage{amsmath}
\usepackage{amssymb}
\usepackage{amsfonts}
\usepackage{subfigure}
\usepackage{algorithm2e}

\onlineid{180}

\vgtccategory{Research}

\vgtcinsertpkg

\title{Improved Optimal and Approximate Power Graph Compression\\
for Clearer Visualisation of Dense Graphs}

\author{Tim Dwyer\thanks{e-mail: tim.dwyer@monash.edu}%
\and Christopher Mears\thanks{e-mail:chris.mears@monash.edu}%
\and Kerri Morgan\thanks{e-mail:kerri.morgan@monash.edu}%
\and Todd Niven\thanks{e-mail:todd.niven@monash.edu}%
\and Kim Marriott\thanks{e-mail:kim.marriott@monash.edu}%
\and Mark Wallace\thanks{e-mail:mark.wallace@monash.edu}
}
\affiliation{\scriptsize Monash University, Australia}

\teaser{
\subfigure[Flat graph with 30 nodes and 65 links.]{
 \includegraphics[width=0.24\textwidth]{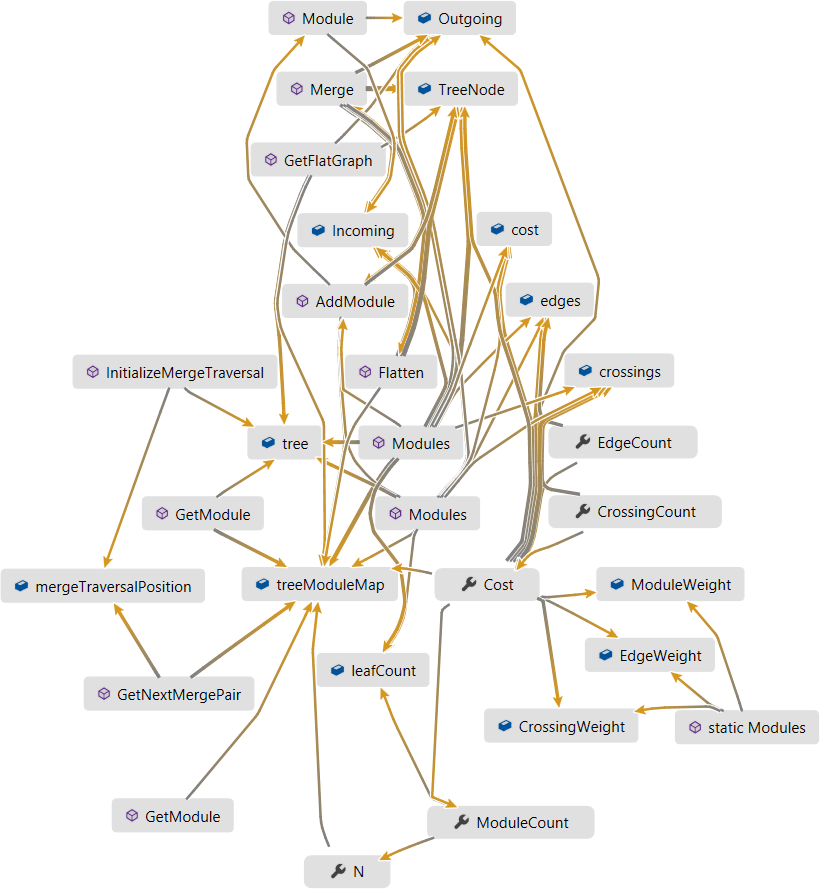}
 }~~~
\subfigure[Power Graph rendering computed using the heuristic of Royer \emph{et al.}~\cite{royer2008unraveling}: 7 modules, 36 links.]{
 \includegraphics[width=0.31\textwidth]{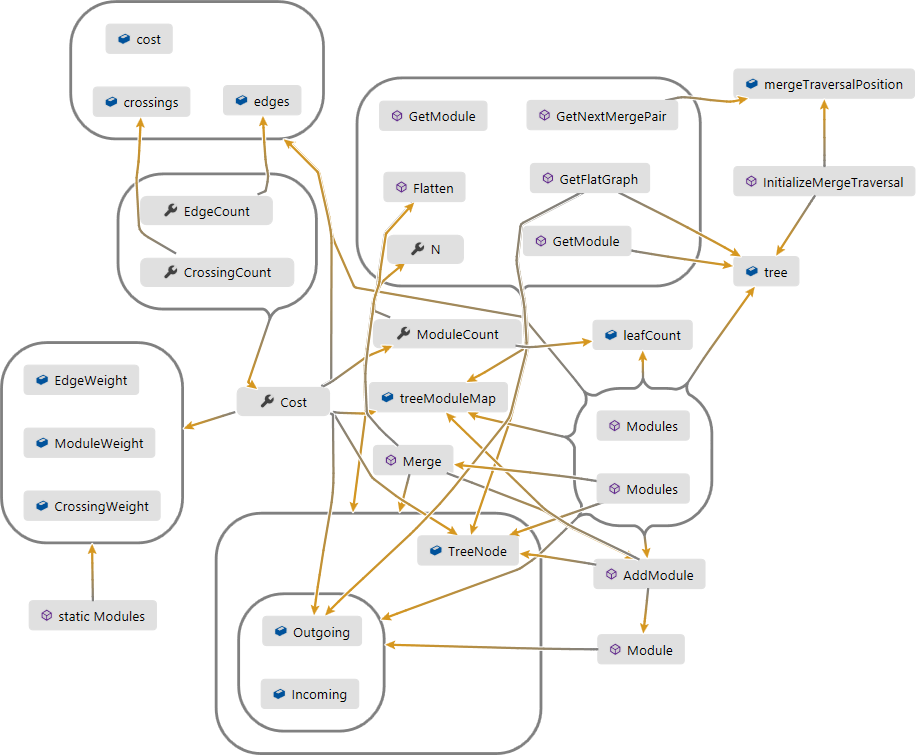}
 }~~~
\subfigure[Power Graph computed using Beam Search (see \ref{sec:beamsearch}) finds 15 modules to reduce the link count to 25.]{
 \includegraphics[width=0.4\textwidth]{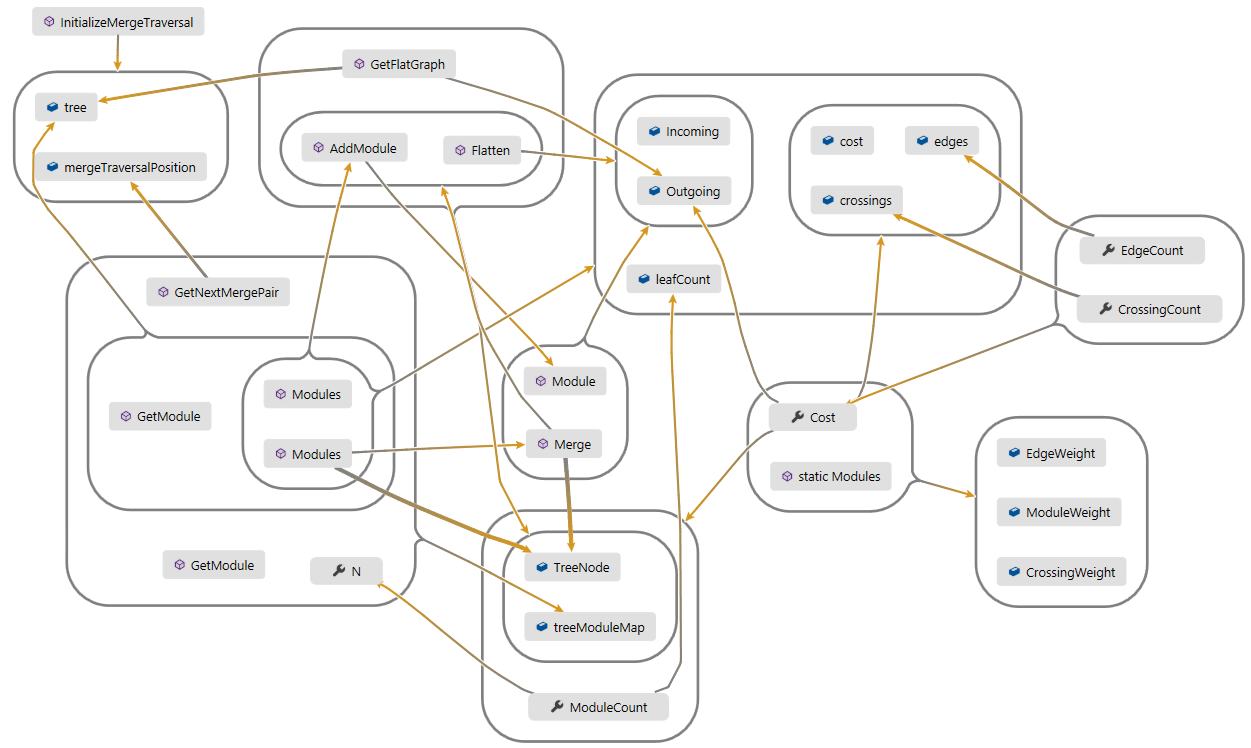}
 }
 \caption{Three renderings of a network of dependencies between methods, properties and fields in a software system.  In the Power Graph renderings an edge between a node and a module implies the node is connected to every member of the module.  An edge between two modules implies a bipartite clique.  In this way the Power Graph shows the precise connectivity of the directed graph but with much less clutter.
 \label{fig:codedependencies}}
}

\abstract{Drawings of highly connected (dense) graphs can be very difficult to read. Power Graph Analysis offers an alternate way to draw a graph in which sets of nodes with common neighbours are shown grouped into \emph{modules}.  An edge connected to the module then implies a connection to each member of the module.  Thus, the entire graph may be represented with much less clutter and without loss of detail.  A recent experimental study has shown that such lossless compression of dense graphs makes it easier to follow paths.  However, computing optimal power graphs is difficult.  In this paper, we show that computing the optimal power-graph with only one module is NP-hard and therefore likely NP-hard in the general case. We give an ILP model for power graph computation and discuss why ILP and CP techniques are poorly suited to the problem.  Instead, we are able to find optimal solutions much more quickly using a custom search method.  We also show how to restrict this type of search to allow only limited back-tracking to provide a heuristic that has better speed and better results than previously known heuristics. 
} 

\begin{document}


\firstsection{Introduction}

\maketitle

In real-world applications such as biology and software engineering it is common to find network structures that are too dense to visualise in a way that individual links can still be followed.  Such graphs occur frequently in nature as power-law or small-world networks.
In practice, very dense graphs are often visualised in a way that focuses less on high-fidelity readability of edges and more on highlighting highly-connected nodes or clusters of nodes through techniques such as force-directed layout or abstraction determined by community detection~\cite{newman2006modularity}.  Dense edge clutter may be alleviated to better show node labels by rendering the edges very faintly or with aggregate techniques like bundling~\cite{Holten2009Bundling}.  Though such approaches may give a rough indication of the general graph structure they make following precise edge paths difficult or impossible.  

Such path following is even more difficult when graphs are directed.  Distinguishing direction on edges allows for up to twice as many distinct edges in the graph.  For people trying to understand the graphs to accurately follow directed paths, edge curves must be drawn in enough isolation that any indications of direction (such as tapering, arrowheads or gradients~\cite{holten2011extended}) are clearly visible.


Recently, alternative approaches have been suggested that attempt to retain the fidelity of individual edge paths by introducing drawing conventions that allow a large number of actual edges to be precisely implied by a small set of composite edges.  In particular, so-called 
Power Graph Analysis constructs a hierarchy over nodes, such that nodes with similar neighbour sets are placed in the same group or \emph{module}.  An edge connected to the module then implies
a connection to each member of the module. Power Graph Analysis utilises \emph{lossless compression} since---by contrast with bundling or community-based clustering---no information is lost in the rendering.  That is, the technique can be said to be \emph{information faithful}~\cite{nguyen2013faithfulness} such that the full graph can be reconstructed by careful inspection of the drawing.  Figure \ref{fig:codedependencies} gives a small example of the application of Power Graph Analysis techniques to a small software dependency graph.

Power Graph Analysis has been shown to have practical application to visualising biological networks~\cite{royer2008unraveling}, detecting communities in social and biological networks~\cite{Varlamis2012MiningPotential} and more recently in software dependency diagrams~\cite{dwyer2013EdgeCompression}. A recent user study has shown that---for path-finding tasks---power-graph style groupings were more readable than flat graphs, even for people with very little training~\cite{dwyer2013EdgeCompression}. 

Relatively little attention has been devoted to algorithms for finding power graph decompositions, i.e. the best choice of modules in the power graph. Royer \emph{et al.}~\cite{royer2008unraveling} give a heuristic for finding the decomposition and Dwyer \emph{et al.}~\cite{dwyer2013EdgeCompression} give a constraint programming formulation for finding the optimal decomposition with respect to various criteria, such as fewer edges or fewer edges crossing group boundaries. Unfortunately, an empirical evaluation of the two algorithms given in~\cite{dwyer2013EdgeCompression}  reveals that the heuristic of ~\cite{royer2008unraveling} is not very effective at finding an optimal solution while the constraint programming approach is too slow for practical use, taking days to run for larger graphs. 
Thus, we need better methods to find power graph decompositions.  This is the subject of this paper.  In particular, we focus on finding power-graphs that are optimal or approximately optimal with respect to the number of edges in the decomposition.

In Section \ref{sec:complexity} we prove that the problem of minimising the number of edges in a power graph decomposition with only one group is NP-hard.  This is strong evidence that the general problem is NP-hard, though specific proofs for unrestricted modules and different goal criteria are needed.

In Section \ref{sec:declarative} we look at declarative models of the Power Graph problem for input into general-purpose solvers.  In particular, we offer some refinements to the Constraint Programming model introduced in \cite{dwyer2013EdgeCompression} (Section \ref{sec:cp}) and propose a new Integer Linear Programming model (Section \ref{sec:ilp}).

In Section \ref{sec:search} we explore explicit search methods.  In Section \ref{sec:beamsearch} we introduce a beam-search based approximate method that produces power graphs for a given flat graph, that are much closer to optimal than previous greedy heuristics.  In Section \ref{sec:optimalsearch} we extend this method to a full-backtracking search strategy that is able to find optimal power graphs relatively efficiently using a lower-bound calculation to cut branches of the search that are not useful.  Our experiments (Section \ref{sec:experiments}) show that this approach is orders of magnitude faster than solving the declarative models using generic solvers.

\section{Background and Definitions}
\label{sec:background}
\begin{figure}
\centering
\subfigure[A small graph with 23 edges.]{
	\includegraphics[width=0.7\linewidth]{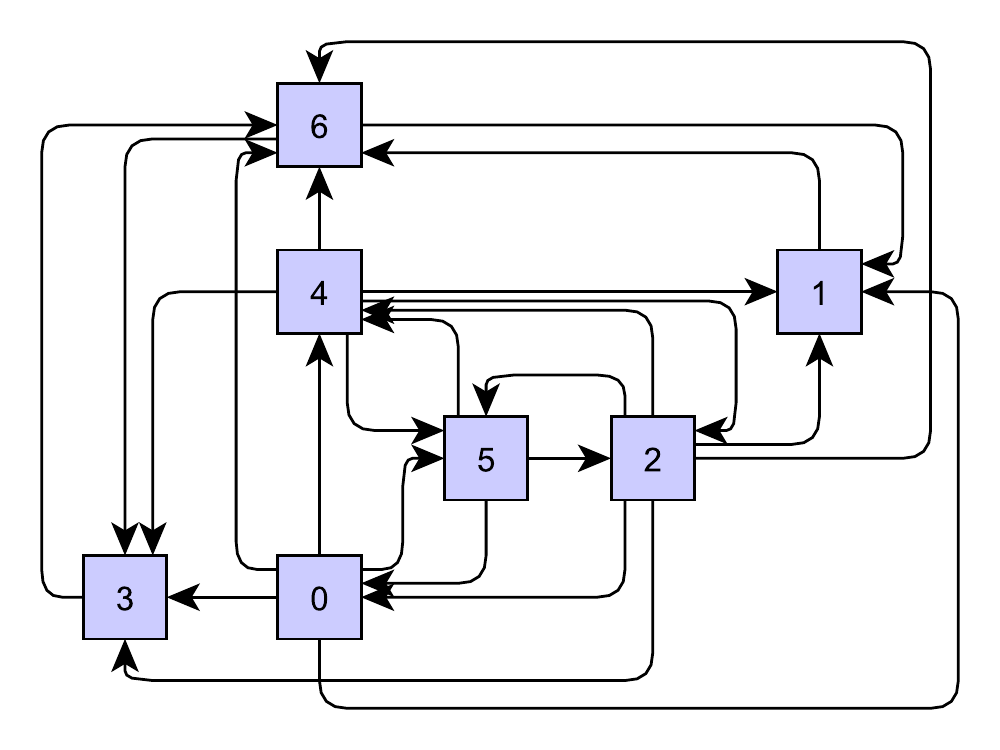}
}
\subfigure[A power graph decomposition of the same graph can be drawn
with only 7 edges.]{
  \label{fig:sevennodes:b}
	\includegraphics[width=0.7\linewidth]{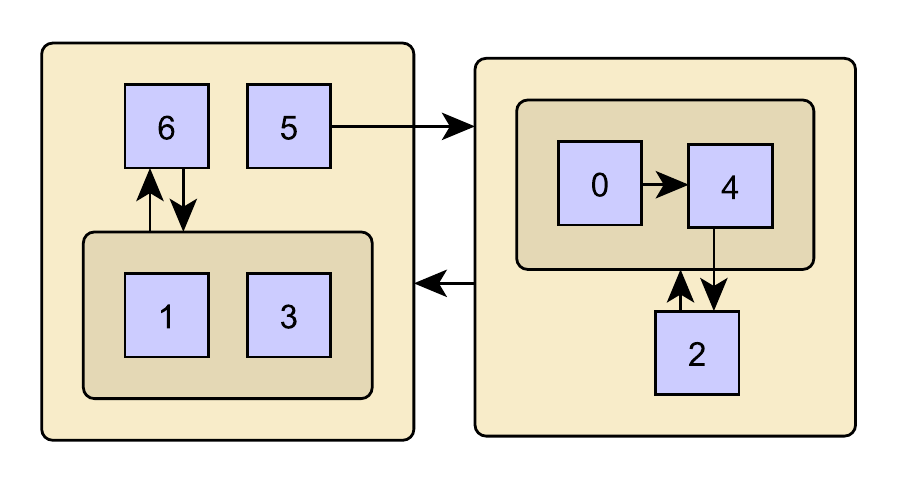}
}
\caption{A small example of power-graph decomposition.
\label{fig:sevennodes}}
\end{figure}

Royer \emph{et al.}~\cite{royer2008unraveling} coined the term Power
Graph Analysis to describe a technique for reducing edge clutter by
introducing a module hierarchy over an undirected graph, such that
edges connected to modules imply a connection to every child of the
module.  More formally, for a directed graph $G=(V,E)$ a power graph
configuration is a set of modules $M$ where each module $m\in M$ is a
subset of $V$.  We assume that $M$ includes the set of trivial modules
$\{\{v\} \mid v\in V\}$.
Our use of the term \emph{module} is due to the similarity with
modules considered in \emph{Modular Decomposition}.  
Like Modular Decomposition, modules can be nested to form a
hierarchy and if two modules overlap one must be
fully contained in the other.
That is, for $m,n \in M$, if $m\bigcap n \neq \emptyset$ then $n
\subset m$ or $m \subset n$.
As in Modular Decomposition,
a power graph decomposition is drawn with a set of
\emph{representative edges} $R$ uniquely specified by $E$ and $M$,
such that a representative edge between two modules represents a
complete bipartite graph in $E$ over the children of the two modules.
That is, for two modules $m$ and $n$, a representative edge $(m,n)\in
R$ represents the set of edges $\mathit{rep}(m,n) = \{(u,v) \mid u\in m, v \in n\}$.
  In drawings of power graphs, the set of
representative edges $R$ must be minimal, i.e. $\nexists e_1, e_2 \in
R$ s.t.\ $\mathit{rep}(e_1) \subset \mathit{rep}(e_2)$. 

Unlike Modular Decomposition, however, the power graph definition
above allows a representative edge to span module boundaries.  For
example, in Figure \ref{fig:sevennodes:b}, the edge from node 5 to module $\{0,2,4\}$ implies that $E$ contains edges $(5,0), (5,2), (5,4)$ but says nothing about any modules containing node 5.  This permits greater compression than modular decomposition,
but can also cause confusion when reading the diagram, as considered
by Dwyer \emph{et al.}~\cite{dwyer2013EdgeCompression}.

Royer \emph{et al.} introduced a simple heuristic for obtaining the
power graph decomposition for a given graph in $O(|V|^3|E|\log|E|)$ time.  
Their procedure involved first computing a
``candidate module hierarchy'' over the nodes in the graph by repeated
greedy matching of nodes or modules with similar \emph{Jaccard index}
between their neighbour sets.  A second iterative phase involves
greedily instantiating the module that reduces the most edges.  This
is repeated until no more edge reduction is possible.

This algorithm is reasonably fast but---as we show in Section
\ref{sec:experiments}---produces results that are far from optimal.  The
Beam Search method described in Section \ref{sec:beamsearch} has faster run-time 
(depending on a beam size parameter) and finds solutions that are very close
to optimal.  Furthermore, in Section \ref{sec:optimalsearch}, we show
that it is easy to introduce backtracking into this explicit search
method, to produce a search that returns an optimal solution.  Such a
complete search is orders of magnitude faster than the best we have
been able to achieve with generic solvers applied to declarative
models, despite extensive experimentation with redundant constraints
to limit the search space in those models, as described in Section
\ref{sec:declarative}.

 
\section{Complexity Analysis}
\label{sec:complexity}
In this section we show that computing the optimal power-graph with a
single module is NP-complete for undirected bipartite graphs. It then
follows that the result holds for directed graphs in general. As an
optimal module for a single module case need not be a module in the
optimal power-graph with multiple modules (see
Figure~\ref{fig:2-modules}), it seems likely that it is hard in the
general case too.  Finding an optimal single module is equivalent to
finding a biclique subgraph $(A,B)$ (here $A$ and $B$ denote the two
independent sets of nodes) that maximises the \textit{edge
  savings}, that is, $|A||B|-|A|$, $|A|\leq |B|$.  The largest part,
$B$, is taken to be the single module.  This problem is different to
the Maximum Vertex Biclique Problem (MVBP) that finds an induced
biclique that maximises the number of nodes, $|A|+|B|$, and to the
Maximum Edge Biclique Problem (MBP) that finds a biclique subgraph
that maximises the number of edges, $|A||B|$.  For example, a biclique
with $|A_1|=3$ and $|B_1|=4$ has more edges than a biclique with
$|A_2|=1$ and $|B_2|=11$, but $B_2$ gives the maximum edge savings.
Similarly, a biclique with $|A_1|=1$ and $|B_1|=9$ has more nodes
than a biclique with $|A_2|=3$ and $|B_2|=6$, but $B_2$ give the
maximum edge savings.  Although MVBP can be solved in polynomial time
for bipartite graphs \cite[Comments on
GT24]{Garey1979}\cite{yannakakis1981}, it is NP-complete for general
graphs \cite{yannakakis1978}.  MBP was shown to be NP-Complete in
\cite{peeters2003Bicliques} using a reduction based on the reduction
used to show the NP-completeness of the Balanced Complete Bipartite
Subgraph problem in \cite{Johnson1987}.  We use this reduction to show
that finding the optimal single module is NP-complete.

\begin{figure}
\centering
\includegraphics[width=\linewidth]{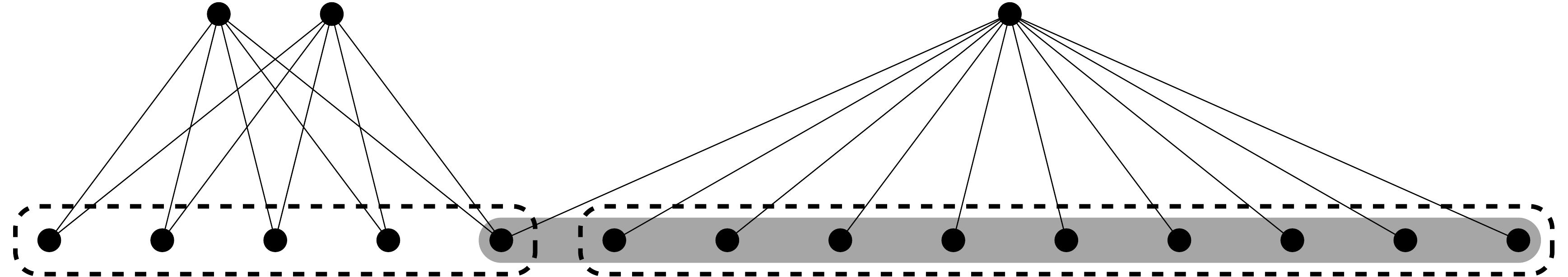}
\caption{The optimal solution for a single module (shaded) is not in
  the optimal solution for two modules (dashed lines). 
\label{fig:2-modules}
}
\end{figure}

\begin{figure}
\centering
\includegraphics[width=\linewidth]{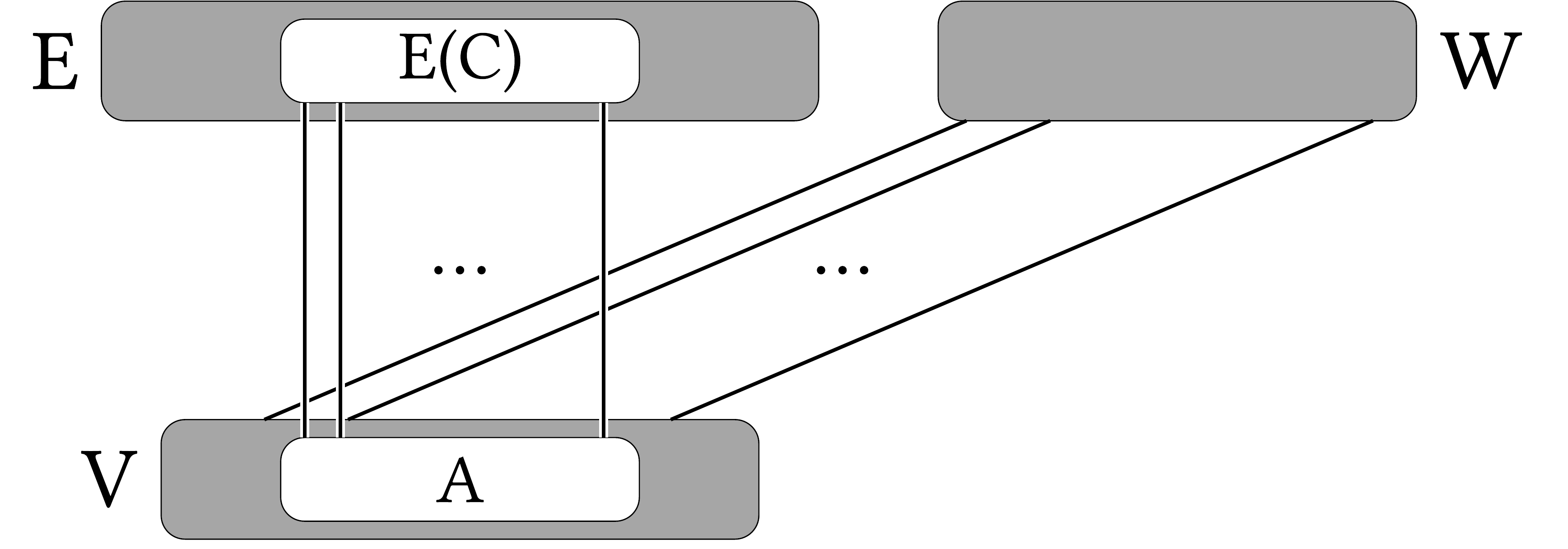}
\caption{The construction $G'$. 
\label{fig:G'}
}
\end{figure}

We define the problem MSMP (Maximum Single Module Problem) as follows:
Given a bipartite graph $G=(V_{1}, V_{2}, E)$ and $k\in \mathbb{N}$,
does $G$ contain a module that achieves an edge saving $\geq k$.
\begin{theorem}
MSMP is NP-complete.
\end{theorem}
\begin{proof}
  It is clear that MSMP is in NP.  We reduce from CLIQUE~\cite{Garey1979} to MSMP using
  the transformation given in \cite{peeters2003Bicliques}.  Let
  $G=(V,E)$ and $k\in \mathbb{N}$ be an instance of CLIQUE.  Without
  loss of generality we assume $k=|V|/2\geq 5$.  We construct a
  bipartite graph $G'=(V_{1}, V_{2}, E')$ (see Figure~\ref{fig:G'}) such that $G'$ contains a
  module with edge savings $\geq k^{3}-k^{2}-k$ if and only if $G$
  contains a clique of size $k$.  Let $V_{1}=V$ and $V_{2}=E\cup W$
  where $W$ is a set of $\binom{k}{2}$ new nodes and $E'=\lbrace
  \lbrace v,e\rbrace: v\in V, e\in E \text{ and } v \text{ not
    incident to }e\rbrace \cup \lbrace \lbrace v,w\rbrace: v\in V,
  w\in W\rbrace$.  It is clear that this construction can be performed
  in polynomial time.

  If $G$ contains a clique $C$ of size $k$ then $(A,B)$ where $A=V-C$
  and $B=W \cup E(C)$ is a biclique subgraph
  of $G'$.  For $k\geq 5$, $|B|=2\binom{k}{2}\geq |A|=k$.  Selecting
  $B$ as a module gives an edge savings of $|A||B|-|A|=
  k^{3}-k^{2}-k$.  

  Suppose $G$ has no clique of size $\geq k$.  We show that any module
  in $G'$ gives a saving of $<k'=k^{3}-k^{2}-k$.  Let $(A,B)$ be a
  maximal biclique in $G'$, $A\subseteq V_{1}$ and $B\subseteq V_{2}$.
  We note that any maximal biclique has $W\subseteq B$ and as $|A|\leq
  |V|=2k\leq |W|\leq |B|$ when $k\geq 5$, $B$ would be the optimal module.

  Now $B=W\cup B'$ where $B'$ corresponds with the edges in $G$ whose
  endpoints are not in $A$.  So $|B'|\leq (|V|-|A|)(|V|-|A|-1)/2 =
  (2k-|A|)(2k-|A|-1)/2$.

  There are two cases.
  \paragraph{Case 1 $|A|>k$:} Suppose $|A|=k+c$, $c\in (0,k]$.  Then
  \begin{align*}
    |A||B|-|A| &\leq (k+c)\left( \frac{k(k-1)}{2} +\frac{(k-c)(k-c-1)}{2}\right)-(k+c) \\
    &= k^{3}-k^{2}-k-\frac{k}{2}\left(c^{2}+c\right)+\frac{c}{2}\left(c^{2}+c-2\right)\\
    &<k^{3}-k^{2}-k.
  \end{align*}

  \paragraph{Case 2 $|A|\leq k$:} Suppose $|A|=k-c$, $c\in [0,k]$.
  Now as $G$ has no clique of size $\geq k$, the number of edges in
  the subgraph of $G$ induced by $V-A$ (and hence $|B'|$) is $\leq
  \frac{(k-2)(k+c)^{2}}{2(k-1)}$ by Tur\'{a}n's
  theorem~\cite{Turan1941}.  So
\begin{align*}
  |A||B|-|A| \leq & (k-c)\left( \frac{k(k-1)}{2} + \frac{(k-2)(k+c)^{2}}{2(k-1)}\right)-(k-c)\\
  \leq& k^{3}-k^{2}-k \\
  &- \frac{1}{2(k-1)} \left( (c^{2}+1)k^{2}-2c^{2}k+c(c^{2}-1)(k-2) \right)\\
  < &k^{3}-k^{2}-k
\end{align*}
for $k\geq 5$.  \hfill $\Box$
\end{proof}


\section{Declarative Models}
\label{sec:declarative}
In this section we investigate how two standard generic methods for solving combinatorial optimisation problems, constraint programming and integer liner programming (ILP), can be applied to
solve the optimal power graph decomposition problem. The advantage of using such generic approaches is that they allow the model to be relatively easily specified in a declarative language such as MiniZinc~\footnote{\url{http://minizinc.org/}} and then run with powerful state-of-the-art solvers.

\subsection{Constraint Programming}\label{sec:cp}

Our starting point is the constraint programming formulation for finding the optimal power graph decomposition given in \cite{dwyer2013EdgeCompression}. 
The input to the model is the number of vertices $nv$, a Boolean array $\mathit{edge}$ which is the adjacency matrix for the edges in the original graph, a limit on the number of modules $ml$ and an upper bound $ub$ on the  objective function. To solve the complete problem we can set $ml = nv$ and $ub = \infty$.

The main decision variables in the problem are the number of modules $\mathit{anm}$ and the Boolean array $\mathit{module[v,m]}$ which gives the vertices $v$ in each module $m$, i.e. 
$\mathit{module}[v,m] \leftrightarrow  v \in m$. 
There are 3 kinds of modules. Modules $1...nv$ are trivial. We constrain these to have a single vertex in them, i.e. module $v$ is the trivial module $\{ v \}$.
Modules $nv+1...anm$ are real modules while modules $anm+1..nm$ are dummy modules where  $nm=nv+ml$. The dummy modules are constrained to be empty.

We require that the modules form a hierarchy: this is enforced by requiring for all modules 
$m$ and $n$, $m \subseteq n \vee n \subseteq m \vee m \cap n = \emptyset$. 
The formulation makes use of  Boolean array $\mathit{mcontains}[m,n]$ which is constrained to hold if module $m$ contains module $n$.

Our objective function is to minimise the number of edges in the power graph.
For any fixed choice of modules there is a unique best choice of edges in the power graph.
We compute for each pair of nodes $m$ and $n$ if there is a \emph{possible} edge between them. There is a possible edge between $m$ and $n$ iff: (1)
for all $u \in m$ and for all $v \in n$, $(u,v) \in E$, and (2) $m=n$ or $m \cap n = \emptyset$.
From this we can compute the \emph{actual} edges in the power graph. This is any possible edge $(m,n)$
which is not \emph{dominated} by some other possible edge $(m',n')$ where $(m',n')$ dominates  $(m,n)$ if $m \subseteq m'$ and $n \subseteq n'$.

We extended the model from \cite{dwyer2013EdgeCompression} with a number of redundant and symmetry breaking constraints that significantly improved its efficiency. \begin{enumerate}
\item
To stop arbitrary re-ordering of the modules we added the symmetry breaking constraint that the real and dummy modules were in decreasing lexicographic order using the standard global function 
$\mathit{lex\_greatereq}$.
\item
We added the redundant constraint that if two vertices have the same ingoing and outgoing edges then they should be in exactly the same set of modules: clearly this is true in an optimal solution.
\item
We added a redundant formulation of module containment  based on the observation that the scalar product $sp[m,n]$ of two modules $m$ and $n$ is $|m|$ iff $m \subseteq n$.
%
%
\item
We added a redundant constraint that every module must have at least one potential edge from it: i.e. there is at least one node that has an edge to all of the nodes in the module.
\end{enumerate}

\subsection{Integer Linear Programming (ILP)}
\label{sec:ilp}
In our next approach we explored the use of ILP.
Mathematical programming techniques like ILP can often outperform constraint
programming on particular problems. A disadvantage, when compared to
constraint programming, is it can be difficult to formulate a given
problem as a mathematical program. Here we formulate an integer linear
program that minimizes the number of edges. A more detailed description of the model can be found in the Appendix. \\
\noindent\underline{\textbf{Input and Parameters}}\vspace{1mm}

\noindent $n$  is the number of vertices of the input graph $G$.\\
$V = \{0,1,\dots,n-1\}$ represents the vertices of the input
  graph $G$.\\
$e(u,v)$  represents the edges of the input graph $G$ as an
  incidence matrix. That is, $e(u,v) = 1$ if $(u,v)$ is an edge of $G$
  and $e(u,v)=0$ otherwise.\\
$m$ is the number of modules with at least two elements (we
  consider each singleton vertex to belong to its own module).\\
$M= \{ 0,1,\dots,n+m-1\}$ represents the set of all modules. 

\noindent \underline{\textbf{Integer decision variables}}\vspace{1mm}

\noindent $sav[m_1,m_2]$ the number of edges that may be
removed from $G$ (to then be replaced by a single edge) if the modules
$m_1$ and $m_2$ are added.

\noindent \underline{\textbf{Binary decision variables}}\vspace{1mm}

\noindent $mod[v,m]$ takes the value $1$ if and only if vertex $v$
  belongs to the module $m$.\\
$ind[v,m]$ takes the value $1$ if and only if $(v,u)$ is an
  edge, for all $u$ in the module $m$.\\
$bic[m_1,m_2]$ takes the value $1$ if and only if, for
  every vertex $v\in m_1$ and every vertex $u\in m_2$, the pair
  $(u,v)$ is an edge in $G$.\\
$dis[m_1,m_2]$ takes the value $1$ if and only if the modules $m_1$
  and $m_2$ are disjoint sets of vertices.\\
$sub[m_1,m_2]$ takes the value $1$ if and only if the module
  $m_1$ is a proper subset of the module $m_2$.\\
$mInd[v,m_1,m_2]$ takes value $1$ if and only if $v\in m_1$ and $ind[v,m_2]=1$.\\
$vMod[v,m_1,m_2]$ takes value $1$ if and only if $v\in m_1$ and $v\in m_2$.\\
$sVer[v_1,v_2,m_1,m_2]$ takes value $1$ if and only if
  $(v_1,v_2)$ is an edge with $v_1\in m_1$ and $v_2\in m_2$ and the
  edge $(v_1,v_2)$ can be removed if $m_1$ and $m_2$ are added.\\
$sMod[m_1,m_2]$ takes value $1$ if and only if $sav[m_1,m_2] > 0$.

\noindent \underline{\textbf{Objective}}\vspace{1mm}

\noindent Maximise \[\sum\{sav[m_1,m_2] -
sMod[m_1,m_2] \mid m1,m2\in M,\ m_1\neq m_2\}.\] 

\noindent \underline{\textbf{Constraints}}\vspace{1mm}

\noindent The majority of the constraints in our ILP model are purely to
correctly define binary decision variables and so we have omitted them
for brevity.
\begin{enumerate}
\itemsep0em 
\item $dis[m_1,m_2] + sub[m_1,m_2] + sub[m_2,m_1] = 1$, $m_1\neq m_2\in M$.
\item $\sum\{sVer[v_1,v_2,m_1,m_2] \mid m_1,m_2\in M, m_1\neq m_2 \}\leq 1$,  $v_1,v_2\in~V$.
\item $sav[m_1,m_2]\leq \sum\{sVer[v_1,v_2,m_1,m_2] \mid v_1,v_2\in V, v_1\neq v_2 \}$,  $m_1\neq m_2\in M$.
\item $sMod[m_1,m_2] \leq sav[m_1,m_2]$,  $m_1\neq m_2\in M$.
\item $sav[m_1,m_2] \leq n^2 sMod[m_1,m_2]$, $m_1\neq m_2\in M$.
\item $\sum\limits_{v\in V} mod[v,m_1] =1$,  $m_1\in V$.
\item $mod[m_1,m_1] = 1$,  $m_1\in V$.
\end{enumerate}
Constraint 1 states that for any two distinct modules $m_1$ and $m_2$,
either $m_1$ and $m_2$ are disjoint, or one is a proper subset of the
other. Constraint 2 says that no edge can be counted twice in the
saving calculation. Constraint 3 defines the variables
$sav[m_1,m_2]$. Constraints 4 and 5 defines
$sMod[m_1,m_2]$. Constraints 6 and 7 force each vertex to
be a singleton module.


\section{Explicit Search Methods}
\label{sec:search}


\begin{figure}
\centering
\includegraphics[width=\linewidth]{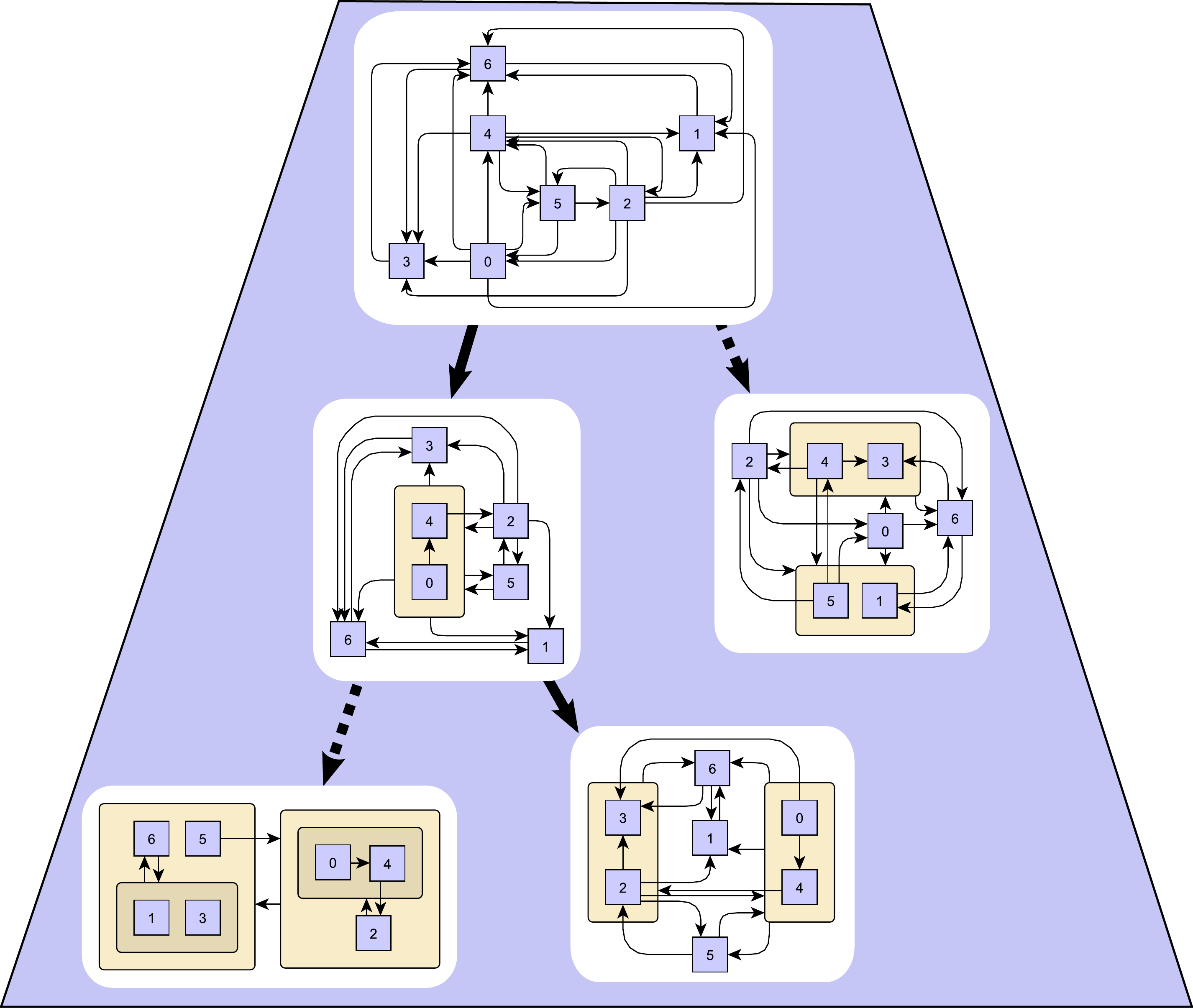}
\caption{The search space of module configurations for a given graph is a tree structure.  Here we highlight several nodes in the very large search space of module configurations for the graph shown in Figure \ref{fig:sevennodes}.  A solid arrow indicates a direct child relationship in the search tree.  That is, the child configuration is obtained from a single merge operation from the parent.  A dotted arrow indicates the target configuration is reachable from the source via multiple merges.
\label{fig:configurationtree}}
\end{figure}

As defined in Section~\ref{sec:background}, for a given flat graph, a \emph{configuration} is a set of modules,
where each module is a set of vertices of the graph and other modules such that the modules
form a hierarchy.  
The objective is to find a configuration that minimises the number of edges.

A configuration can be constructed by adding one module at a time.
The order in which the modules are added has no effect, so we
impose a constraint that only top-level modules are added -- that is,
if the configuration is $\{m_1,m_2,\dots,m_n\}$, then a new module
must not be a subset of any $m_i$.

The set of possible module configurations is then a
tree, with the configuration containing only trivial modules (the flat graph) at the
root, as illustrated in Figure \ref{fig:configurationtree}.  Each node in the tree has one child for each module that can be
added to that configuration.  

To prevent multiple paths to the same
configuration, we impose an arbitrary ordering on the children of a
node, and force that if $m_1$ and $m_2$ are two sibling child modules
such that $m_1$ appears to the left of $m_2$, then the module $m_1$
may not appear in any configuration in the subtree of $m_2$.

Thus, this tree will contain every possible configuration of modules,
and a search for the best configuration may simply do a full traversal
of the tree; however, this is impractically slow.  In the remainder of
this section we describe how to improve the efficiency of this search.

\subsection{Beam Search}
\label{sec:beamsearch}

Using the above definitions, we can define greedy heuristics to obtain
a power-graph configuration with significantly fewer representative
edges than are required to draw the flat graph.

To limit the branching at each search-tree node $t$ associated with a
configuration $C$, we restrict children of $t$ to be those
obtainable through a \emph{merge} of two top-level modules in $C$.
That is, a merge of modules $m$ and $n$ creates a new configuration
$C^\prime = C \cup \{\{v\} \mid v \in m\cup n\}$.

A simple best-first search is---for a given starting configuration
$C$ with $n_t$ top-level modules---to try all possible merges to
give $\{C^\prime_1, \ldots, C^\prime_{n_t (n_t - 1)}\}$ children
and then simply take the one with the fewest representative edges to
be the next configuration.  We repeat until no further improvement is
possible.

Since the set of representative edges $R$ used to draw the power graph
must be minimal, such a merge operation may leave $m$ or $n$ with no
associated representative edges, i.e.\ $\mathit{rep}(m,\_)=\emptyset$
or $\mathit{rep}(\_,n)=\emptyset$.  Since a module with no associated
edge in $R$ serves no purpose, we remove such modules from the merge result
$C^\prime$.

In practice, a full merge is not required just to calculate the reduction in edges achieved by the merge.  Consider a merge of two modules $m$ and $n$ in $C$ with representative edges $R(C)$ giving a configuration $C^\prime$ with edges $R(C^\prime)$.  If the outgoing and incoming neighbour sets in $R$ of $m$ are given by $N^+(m)$ and $N^-(m)$, then the number of edges 
in $R(C^\prime)$ will be precisely: 
$$\mathit{nedges}(m,n) = |R(C)| - |N^+(m)\cap N^+(n)| - |N^-(m)\cap N^-(n)|$$  

Applied to a graph $G=(V,E)$, $O(|V|)$ such greedy merges are possible before
no more improvement is possible.  Since a given module may include all edges
in its neighbour set, computing $\mathit{nedges}(m,n)$ can be $O(|E| \log
|E|)$ and to find the best possible merge this must be computed for all pairs
of modules.  Therefore, a na\"ive implementation could take 
$O(|V|^3 |E| \log |E|)$ time though in practice the number of modules that must be considered and the size of their neighbour sets diminishes quickly.

Beam search \cite[p.~195]{Kaufmann1991Paradigms} gives us a way to introduce some limited backtracking into this best-first strategy.
Instead of only considering the single best merge at each iteration, we maintain a \emph{beam} of the $k$ best solutions found so far.
The full process is shown in Algorithm \ref{alg:beamsearch}.

One detail in this algorithm is the check before insert of a configuration
into the beam that we have not already considered a configuration that is structurally
identical.  This is done by maintaining a hashset $S$ with \emph{signatures} of 
configurations previously inserted into the beam.
For a configuration $C$ the signature $\mathrm{sig}(C)$ is obtained by a canonical (ordered) traversal of the module hierarchy. 

\begin{algorithm}
\SetAlgoLined
\KwData{A starting configuration of modules $C_0$
   with representative edges $R_0$.  A beam size $k\geq 1$.
}
\KwResult{An improved configuration $C^\star$ with $|R^\star|\leq |R_0|$ }
\vspace{2mm}
$\mathsf{beam} \leftarrow$ priority queue of configurations \\
 \Indp\Indp such that first element in queue has maximum $|R|$\;
 \Indm\Indm push$(\mathsf{beam}, C_0)$\;
 $S \leftarrow$ hashset of \emph{signatures} of configurations in $\mathsf{beam}$\;
\Repeat{not $\mathsf{improved}$}{
   $(\mathsf{improved}, \mathsf{currentbeam}) \leftarrow$ (false, copy $\mathsf{beam}$)\;
  \ForEach{$C \in \mathsf{currentbeam}$}{
      $\mathsf{merges}\leftarrow \{(e,m,n) \mid $\\
      \Indp\Indp\Indp $e\leftarrow \mathit{nedges}(m,n), (m,n)\in C\times C\}$\;
      \Indm\Indm\Indm
      $\mathsf{kbest}\leftarrow$ from $\mathsf{merges}$ take $k$ triples $(e,m,n)$ \\
      \Indp\Indp with smallest $e$ where $sig(C, (m, n)) \not\in S$\;
      \Indm\Indm
     \ForEach{$(e,m,n)\in \mathsf{kbest}$}{
       $b \leftarrow$ first (i.e. worst) configuration in $\mathsf{beam}$\; 
       \If{$|\mathsf{beam}|<k$ or $|R(b)|>e$}{
         $C^\prime\leftarrow \mathrm{merge}(m,n)$\;
         push$(\mathsf{beam}, C^\prime)$\;
         insert$(S, \mathrm{sig}(C^\prime))$\;
         $\mathsf{improved} \leftarrow$ true\;
       }
       \If{$|\mathsf{beam}|>k$}{
         pop$(\mathsf{beam})$\;
       }
     }
   }
}
\Return configuration in beam with smallest $|R|$
\vspace{2mm}
\caption{Beam Search
\label{alg:beamsearch}}
\end{algorithm}

\subsection{Optimal Search}
\label{sec:optimalsearch}

We can show certain properties of the configuration tree that will
help us to eliminate regions of the tree without missing any optimal
configurations.  First, we restrict the definition of ``optimal'' to
be a configuration that has no redundant modules.  That is, if two
configurations $C$ and $C \cup \{m\}$ have the same objective cost,
only $C$ can be considered optimal.

\textit{Adding a module cannot increase the number of edges.}  Any new
edge added to the graph must replace some existing edges.  A new edge
between existing module $m$ and new module $n$ replaces all previous
edges between members of $m$ and members of $n$.  Since there must be
at least one such previous edge, the number of edges removed is at
least the number of edges added.

\textit{Every optimal configuration can be reached by a sequence of
  \emph{improving} module additions.}  That is, every module added 
reduces the number of edges.

An outline of the reasoning is as follows.  Let the desired optimal
configuration be $C=\{m_1,m_2,\dots,m_n\}$.  Starting from the empty
configuration, choose any module that contains no sub-modules.  Adding
this module will reduce the number of edges, so add it to the current
configuration.  This module must reduce the number of edges -- if it
didn't, then $C$ wouldn't be optimal.  Then after adding this module,
find another module that contains no sub-modules (or only sub-modules
that have already been added) and add that module, and so on.

\textit{An optimal solution has at most $n-2$ modules, where $n$ is the
  number of vertices in the flat graph.}
A configuration can have at most $n-1$ modules.  A configuration with
$n-1$ modules must include a module that contains all vertices in the
graph.  However, such a module is not an improving module and so
cannot appear in an optimal solution.  Therefore, an optimal
configuration can have at most $n-2$ modules.

\textit{If adding module $m$ to configuration $C$ would remove $e$ edges,
  then adding module $m$ to configuration $C' \supset C$ would
  remove at most $e$ edges.}

Adding a module $m$ with members $\{a,b,\dots\}$ to configuration $C$
can only remove edges by replacing existing edges $\{(a,v), (b,v),
\dots\}$ with a single edge $(m,v)$ (and similar for opposite-oriented
edges).  A configuration $C' \supset C$ has the same set of edges as
$C$, except that some edges are replaced by a smaller set of
edges.  Therefore, the module $m$ can only replace the same set of
edges $\{(a,v), (b,v), \dots\}$ or a smaller set where some of those
edges have been merged.

These properties allow us to calculate a lower bound on the best
possible configuration in under a given configuration in the search
tree.  Given a configuration $C$ with $m$ modules, we can compute the
minimum number of edges 
for any configuration that is a superset of $C$.

Let $c_1,c_2,\dots,c_n$ be the modules permitted to be added to
configuration $C$.  Each of these can be given a score $S(c_i)$ that
is the number of edges removed by adding that module: $S(c_i) = E(C) -
E(C \cup \{c_i\})$.  Since we can add at most $n-2-m$ modules to $C$,
the most edges we can possibly remove from $C$ is the sum of the
best $n-2-m$ scores of the modules $c_1,c_2,\dots,c_n$.


Therefore, for a given configuration $C$ we can calculate the score of
a hypothetical configuration $C' \supset C$ that has the minimum
possible edges.  (The configuration is only hypothetical because it
likely violates the hierarchy restriction on the modules.)
If the
objective cost of $C'$ is no better than the best-known solution, the
search can backtrack immediately from $C$.
An outline of the exhaustive search algorithm is given in Algorithm~\ref{alg:optsearch}.

\begin{algorithm}
  \SetAlgoLined
  \KwData{Current configuration, incumbent configuration, flat graph.}
  \KwResult{The optimal configuration.}
\vspace{2mm}
  \If{current configuration is better than incumbent}{
    record current configuration as incumbent
  }
  $\mathsf{bound} \leftarrow$ compute bound of current configuration\;
  \If{$\mathsf{bound}$ is not better than incumbent}{
    backtrack\;
  }
  \If{current configuration below module limit}{
    $\mathsf{modules} \leftarrow$ modules that can be added to current
    configuration\;
    Sort $\mathsf{modules}$ with most edge-reducing modules first\;
    \For{each module $m \in \mathsf{modules}$}{
      Recurse with $m$ added to current configuration\;
      Upon backtracking, forbid $m$ from $\mathsf{modules}$\;
    }
  }
  \Return incumbent configuration
\vspace{2mm}
  \caption{Outline of optimal search.}
  \label{alg:optsearch}
\end{algorithm}


We can further reduce the number of modules considered at each step of
the search.
We only need to consider \emph{binary} modules---as
produced by the merge operation for the beam search
(Section~\ref{sec:beamsearch})---
which
have exactly two members (which themselves may be either modules or
single vertices).  The search is still guaranteed to find an optimal
configuration.

\textit{Any optimal configuration has a counterpart that has only
  binary modules, and each such binary module reduces the number of
  edges.}
  Consider an optimal configuration $C$, constructed by the sequence
  of modules $<m_1,m_2,\dots,m_n>$.  Let $m_i$ be a non-binary
  module.  We can convert $m_i$ into a binary module by grouping
  together two of its members arbitrarily.  (This newly created
  sub-module will have no edges: if it was beneficial for this new
  sub-module to have edges, $C$ would not have been optimal.)  The new
  sub-module is inserted immediately before $m_i$ in the construction
  sequence.  If $m_i$ is still not binary, recurse.

  All new sub-modules created in this way reduce the number of edges
  during construction.  We know that adding the non-binary module
  $m_i$ reduces the number of edges.  From this we know that the
  members of $m_i$ share at least one common predecessor or common
  successor.  Therefore, any pair of the members of $m_i$ share a
  common predecessor or successor, and adding a module containing just
  that pair will reduce the number of edges.  (Any modules that become
  redundant during this construction can be removed.)

  The consequence of this property is that in the search algorithm,
  the modules to be added at each step need only be these binary
  modules instead of all possible modules.  This greatly reduces the
  search space.



\section{Experiments}
\label{sec:experiments}
\subsection{Heuristic Methods}
\begin{figure*}
\centering
\subfigure[Flat graph with $|V|=10, |E|=51$.  An orthogonal layout style was used (Y-Files implementation of topology-shape metrics approach) in an attempt obtain a layout and routing in which each individual edge path is traceable.  Such traceability is necessary to support discerning reachability or path-finding tasks.]{
\includegraphics[width=0.45\textwidth]{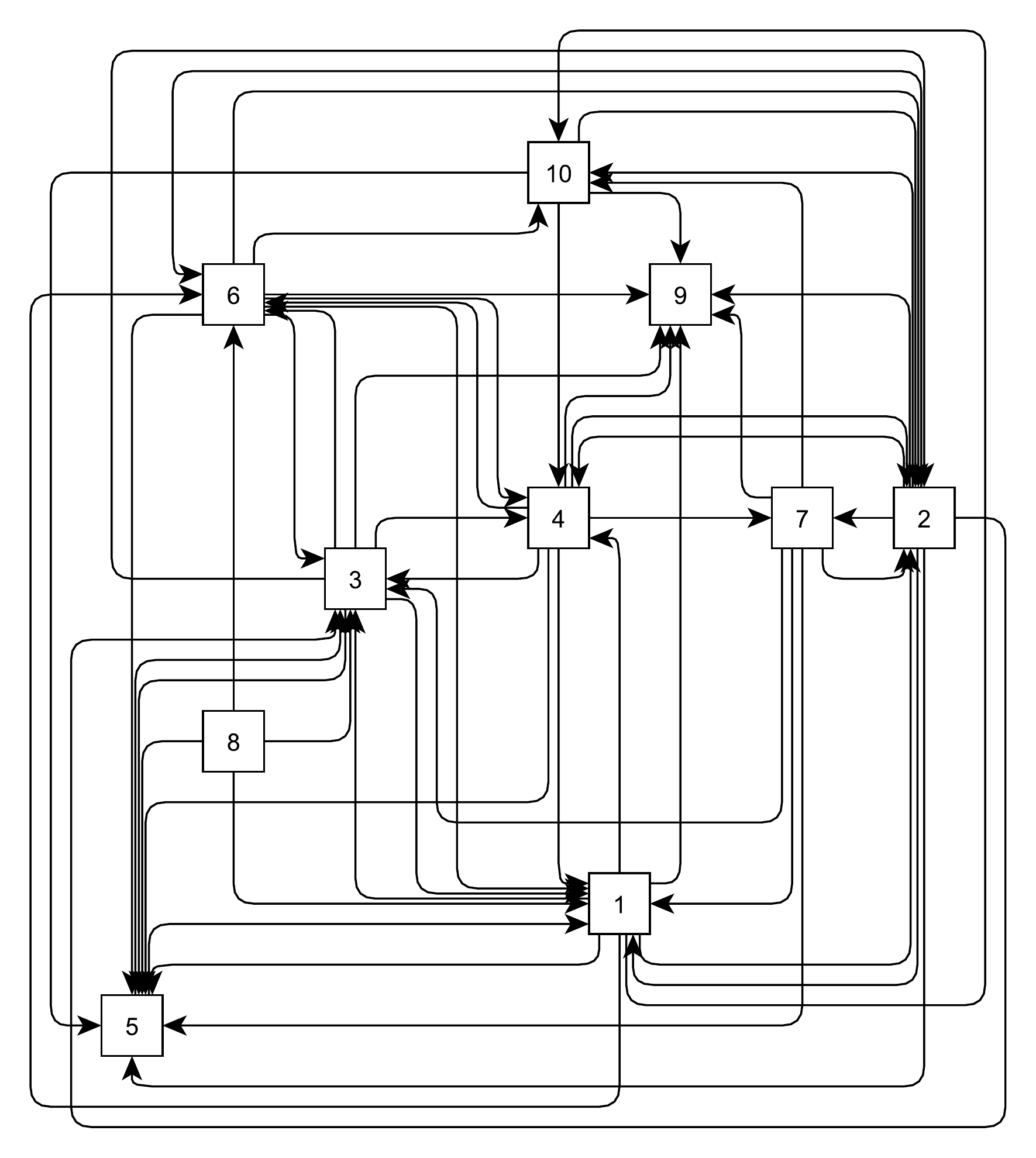}
\label{fig:graph-17-flat}
}~~~
\subfigure[Greedy Jaccard Metric decomposition.  
Decomposition has 5 modules, 23 edges and 21 crossings between edges and module boundaries.]{
\includegraphics[width=0.45\textwidth]{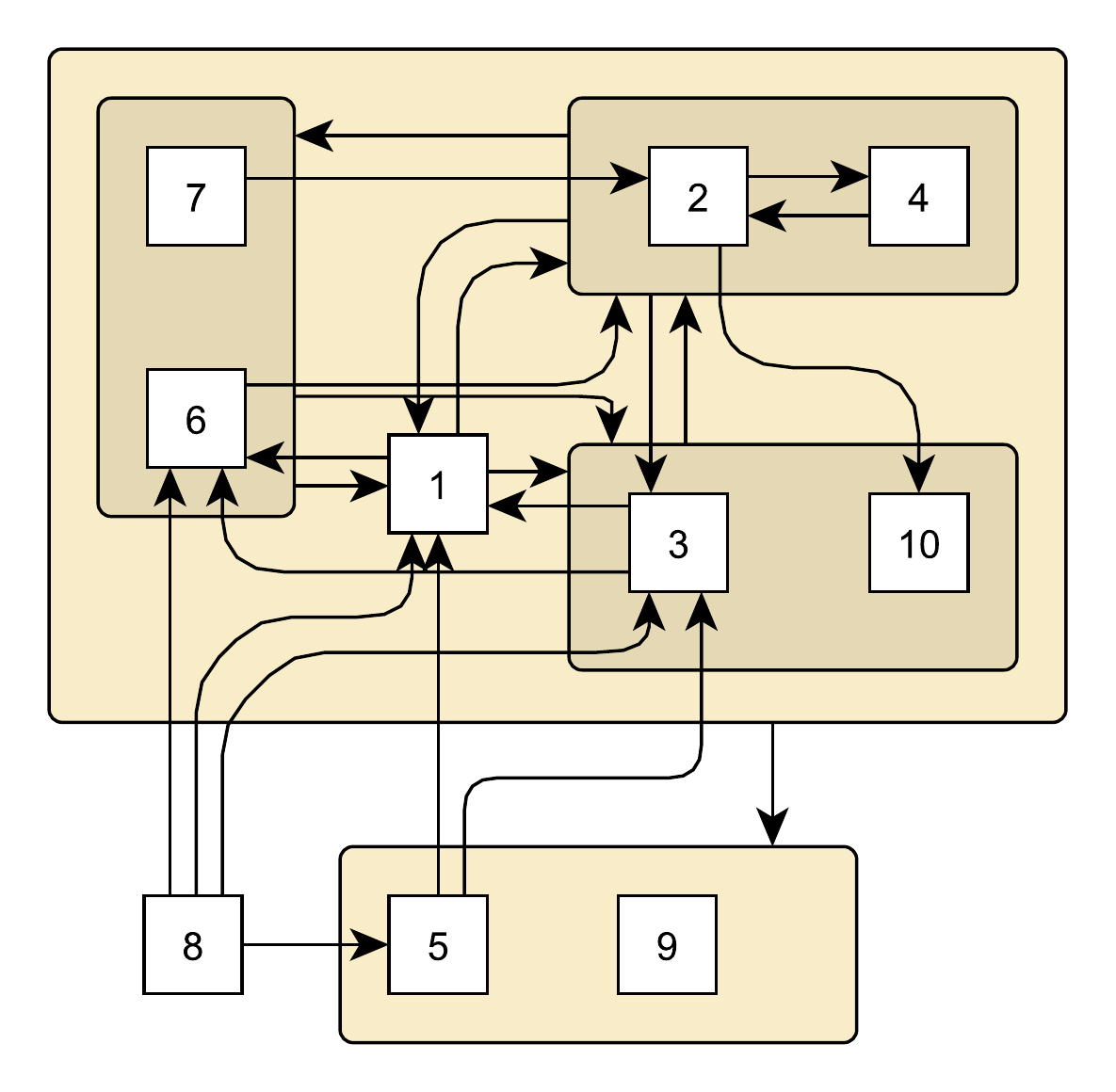}
\label{fig:graph-17-crossingweight1-greedy}
}\\
\subfigure[Beamsearch decomposition $k = 1$.  Decomposition has 6 modules, 19 edges and 17 crossings between edges and module boundaries.]{
\includegraphics[width=0.49\textwidth]{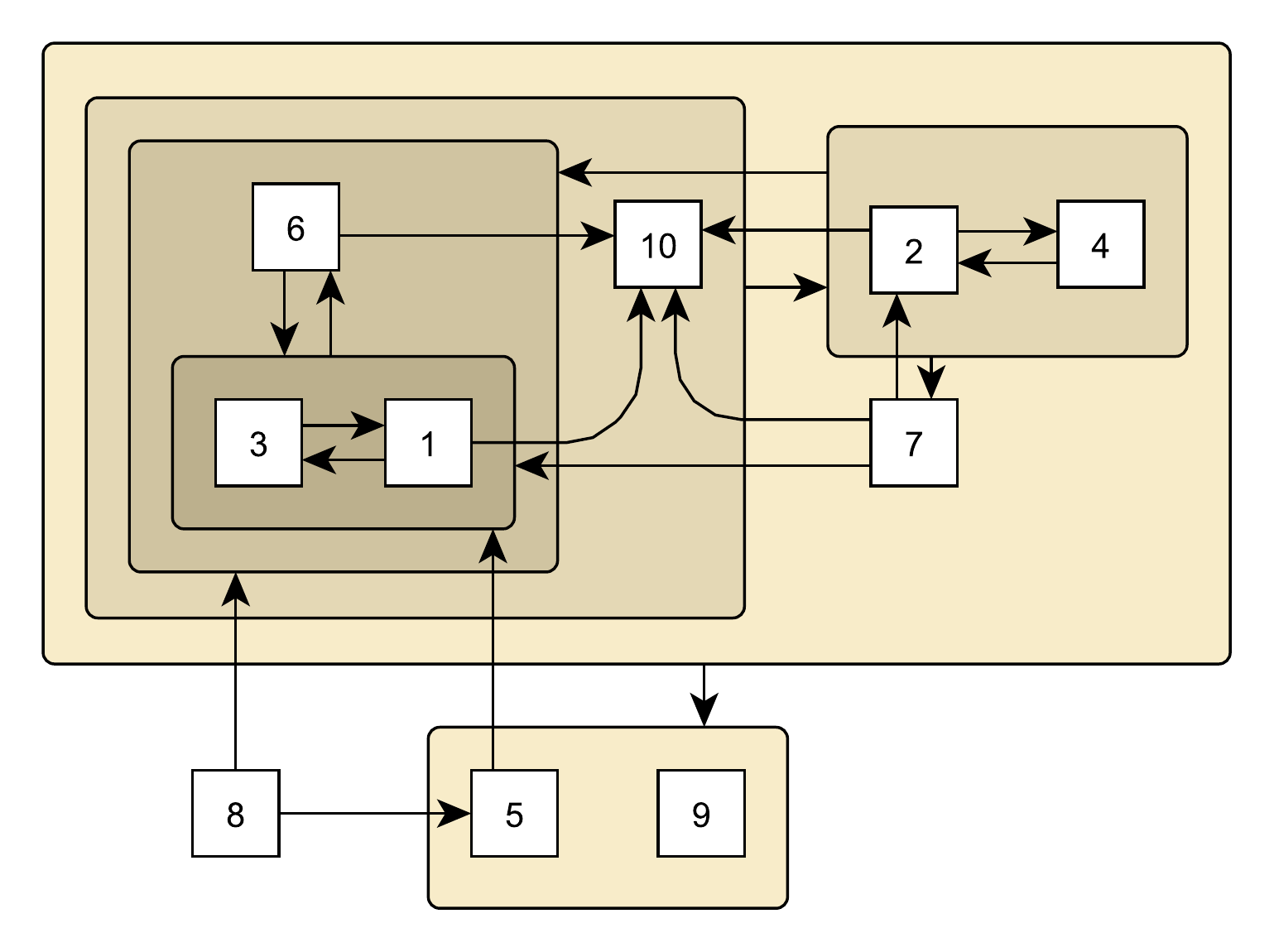}
\label{fig:graph-17-crossingweight1-beamsize1}
}~~~
\subfigure[Optimal decomposition obtained with beam search $k=10$ (optimality proven with method described in Section \ref{sec:optimalsearch}).  
Decomposition has 6 modules, 16 edges and 13 crossings between edges and module boundaries.]{
\includegraphics[width=0.49\textwidth]{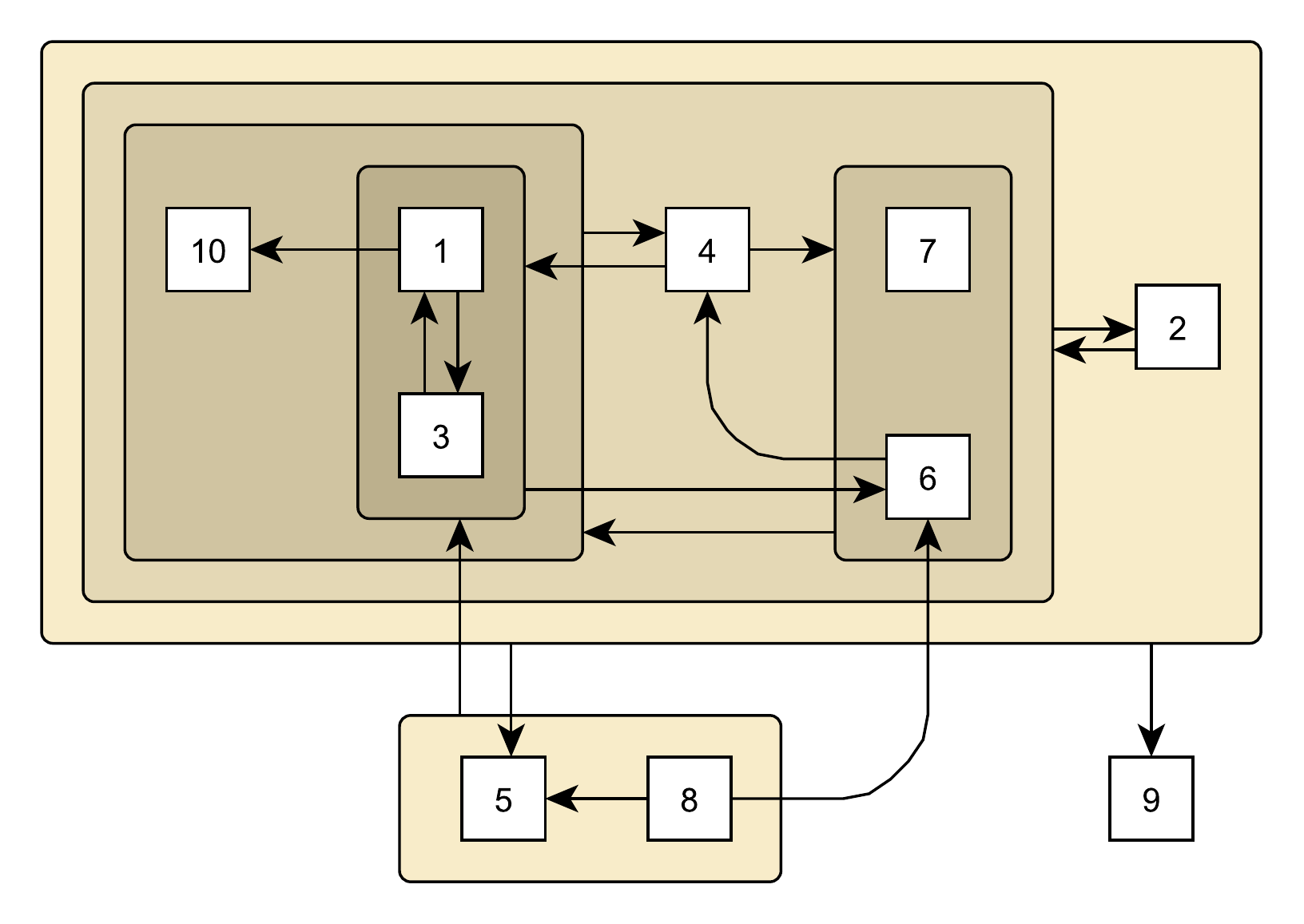}
\label{fig:graph-17-crossingweight1-optimal}
}
\caption{Flat and power graph renderings of a graph with 10 nodes and 51 edges.
It is interesting to note that the reduction in clutter from \subref{fig:graph-17-crossingweight1-greedy} to \subref{fig:graph-17-crossingweight1-beamsize1} is fairly obvious, however, the qualitative improvement from \subref{fig:graph-17-crossingweight1-beamsize1} to \subref{fig:graph-17-crossingweight1-optimal} is less dramatic.
Layouts were obtained using the yEd software (\url{http://yfiles.com}) using a combination of automatic layout and manual refinement with automatic edge routing.
\label{fig:visualcomparison}
}
\end{figure*}
Following Dwyer \emph{et al.}~\cite{dwyer2013EdgeCompression}, we generate random scale-free directed graphs with various numbers of nodes using the model of Bollob\'as \emph{et al.}~\cite{Bollobas2003DirectedScaleFreeGraphs}.  To generate a graph with $|V|$ nodes we control for density such that the number of edges is roughly proportional to $\frac{3}{2} |V|^\frac{3}{2}$.  For example, a graph with 10 nodes will have around 50 edges, while a graph with $|V|=100$ will have around 1500 edges.

All heuristics were implemented in $C\#$ and run on a modern tablet PC\footnote{Intel Ivy Bridge Core i7, up to 2.6GHz}.  The source-code is available under an open-source license\footnote{\url{http://dgmlposterview.codeplex.com}}.

Figure \ref{fig:comparisonchart} compares the decompositions obtained by applying our various heuristic powergraph decomposition methods to graphs generated as described above.  It is clear that the Greedy Jaccard Clustering method~\cite{royer2008unraveling} (JC) is easily beaten in terms of quality of the final solution, even by beam search with a beam size $k$ of 1 (i.e.\ best-first search with no backtracking at all).  Increasing $k$ to 10 does improve the results slightly even---in some cases---returning the optimal (e.g.\ see Figure \ref{fig:graph-17-crossingweight1-optimal}).  However, this is achieved at a ten-fold increase in running time (see Figure \ref{fig:heuristictimings}).  Increasing $k$ further to 100 only occasionally results in a slightly better solution.  Note that with $k$ small compared to the total number of configurations in the search space there is no guarantee that beam search will return the optimal.

Figures \ref{fig:graph-17-crossingweight1-greedy}, \subref{fig:graph-17-crossingweight1-beamsize1} and \subref{fig:graph-17-crossingweight1-optimal} show the renderings of the results obtained for a small instance with 10 nodes and 51 edges.  The quality of the renderings in terms of visual clutter, does seem to reflect the numbers seen in Figure \ref{fig:comparisonchart}.  There is a fairly marked visual improvement between JC (Figure \ref{fig:graph-17-crossingweight1-greedy}) and Beam Search with k=1 ($BS_1$) (Figure \ref{fig:graph-17-crossingweight1-beamsize1}), in particular $BS_1$ can be rendered without edge-edge intersections while for the JC decomposition it is impossible.  However, there is a much less obvious improvement from $BS_1$ to $BS_{10}$ (Figure \ref{fig:graph-17-crossingweight1-optimal}): the latter has three fewer edges and four fewer edge-module boundary crossings.  The result found by $BS_{10}$ is, in fact, the optimal as confirmed by running a full search.

In terms of running time, Beam Search with $k=1$ is easily the fastest method, as seen in Figure \ref{fig:heuristictimings}.  Generally, the running time (and also memory requirements) for beam search grow linearly in $k$.  For example, $BS_{10}$ has very close to twice the run time of $BS_{5}$.  The exception is that the running time of $BS_{1}$ is significantly less than $10\%$ that of $BS_{10}$ since we are able to apply all the merges entirely in place and no copying of configurations is required nor checks for structural equivalence of configurations. 

\begin{figure}
\centering
\includegraphics[width=\linewidth,trim=2cm 1cm 2cm 1cm]{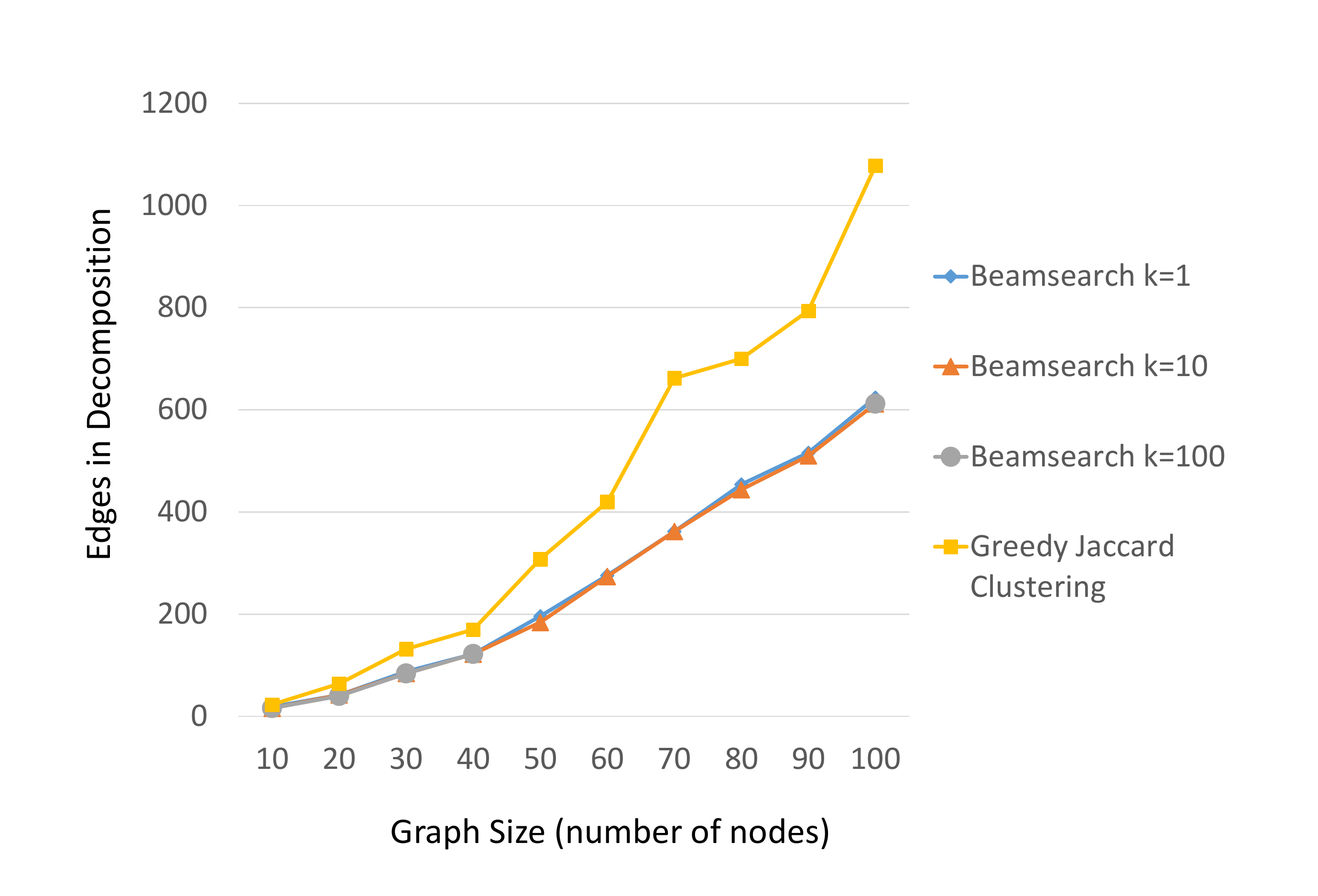}
\caption{Comparison of cost versus graph size for the various power graph decomposition heuristics.
\label{fig:comparisonchart}
}
\end{figure}

\begin{figure}
\centering
\includegraphics[width=\linewidth,trim=3cm 1cm 1cm 1cm]{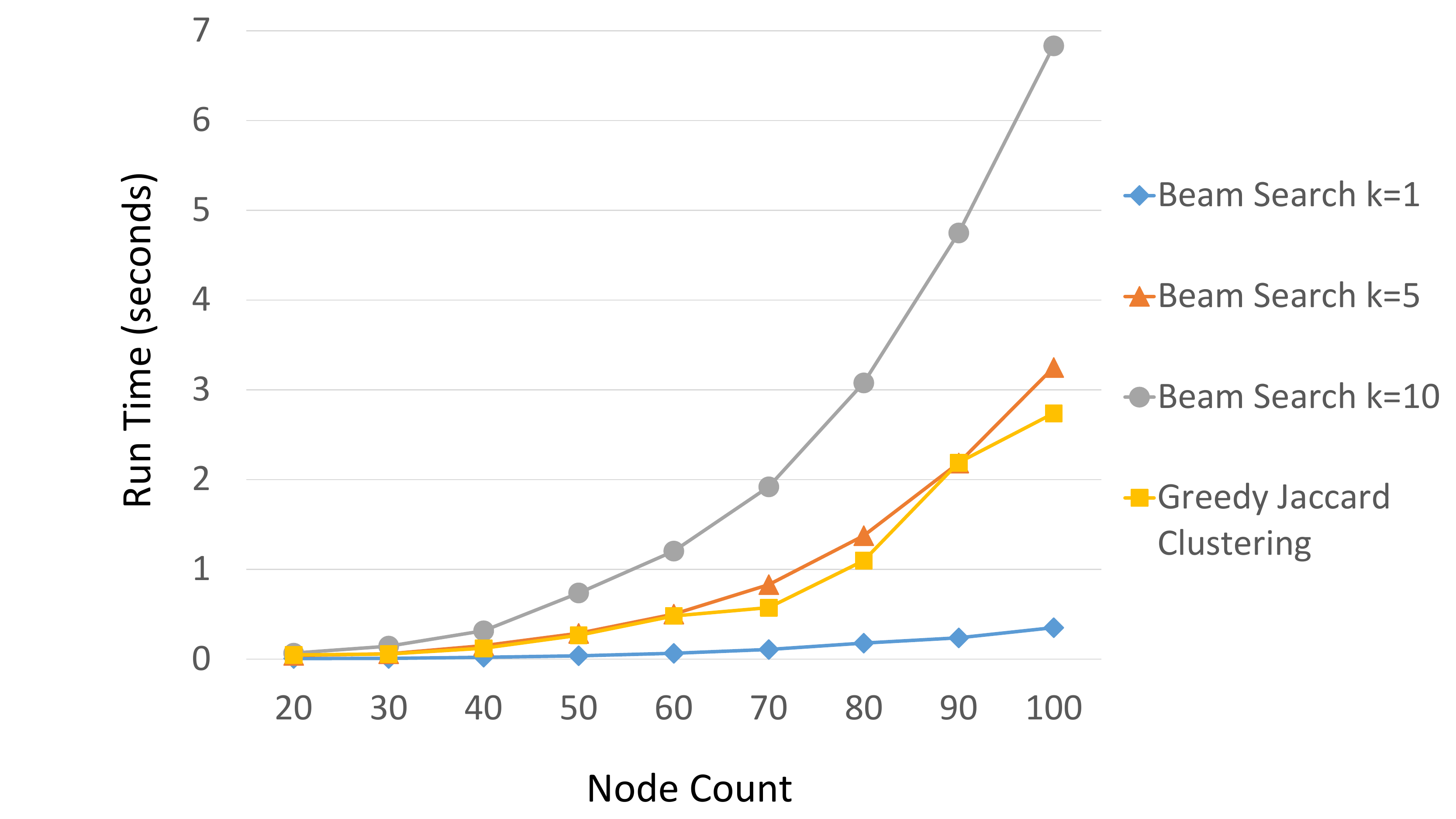}
\caption{Comparison of running times for Beam Search with various beam sizes and Greedy Jaccard Clustering heuristic.
\label{fig:heuristictimings}
}
\end{figure}

\subsection{Optimal Methods}
We have compared the efficiency of the exhaustive search methods
across a corpus of 10-vertex graphs.  The running times on a single
representative graph are given below.

\vspace{2mm}
\begin{tabular}{rr}
  ILP model & $> 160$ CPU hours \\
  CP model without redundant constraints & $> 72$ hours\\
  CP model with redundant constraints & $26.6$ hours \\
  Optimal Search & $8$ minutes \\
  Optimal Search (edge-minimization only) & $22$ seconds \\
\end{tabular}
\vspace{2mm}

Note that the ILP model only tries to minimize edges, while both
variants of the CP model try to break ties by minimising the number of
module boundary crossings and number of modules as in
\cite{dwyer2013EdgeCompression}.
For comparison with the CP model, we implemented two variants of the
optimal search: one only minimising edges, and the other breaking ties
in the same way as the CP model.  The tie-breaking variant is slower
because it explores a larger search space and must do more work at
each search step to compute boundary crossings.

All exhaustive methods were run on an Intel Core i7 2.67 GHz CPU.
The
ILP solver was run in parallel, with 8 cores running for 20 hours.
The ILP model was run with the commercial Gurobi\footnote{\url{http://www.gurobi.com/}}
solver, and the CP
model was run with the lazy-clause generation solver CPX\footnote{\url{http://www.opturion.com/cpx.html}}.
The ILP model and CP model without redundant constraints failed to
find the optimal solution.  The other three methods found the optimal
solution and proved it to be optimal.


\section{Conclusion}
\label{sec:conclusion}
This paper has presented a number of results, both practical and theoretical.  On the practical side, the best-first search method (or beam search $k=1$ or $BS_1$) is much faster than the previous best available heuristic, the Greedy Jaccard Clustering \cite{royer2008unraveling} (JC). 
For example, a dense graph with 100 nodes and 1500 edges can be decomposed in 0.35 seconds with $BS_1$ compared to 2.74 seconds for JC.  This is a significant enough performance improvement to enable new scenarios like continuous decomposition of live-streaming dynamic graph data.

Our experiments have shown that it also computes decompositions that are much closer to optimal, e.g.\ for the graph above $BS_1$ achieves a decomposition with 624 edges compared to 1078 for JC.  We have also shown the applicability of the beam search technique to obtain still more optimal results (e.g.\ again for the 100 node graph $BS_{10}$ finds a 612 edge decomposition) in reasonable time and memory.

For the optimal power graph decomposition problem, we have contributed the first known ILP model and significant improvements to a previous CP model.  While there have been many improvements to the compiler and solver technologies for such mathematical programming techniques this problem still defies simple, efficient declarative models.  By contrast, we have provided an explicit search method that is able to find and prove optimality in minutes compared to days.  Truly optimal solutions may not be necessary for practical power graph visualisation, but being able to efficiently compute them at least for small instances has proven invaluable in helping us develop better heuristic methods such as the beam search.  For example, development of the beam search was in large part motivated by discovering just how bad the Jaccard clustering method was by comparison to optimal decompositions.

On the theoretical side we give the first known NP-hardness proof for a power-graph analysis problem (the one module case) which is strong ground-work for a general NP-hardness proof, though it is also likely that variations of the problem (such as slightly different constraints or goal functions) may require separate proofs or may even be polynomial time.

\section{Further work}

As mentioned above more complexity analysis is required for the general power graph problem and its variants.

A related technique for simplifying dense graphs is confluent graph drawing \cite{dickerson2005ConfluentDrawings}.  Like power-graph decomposition this also involves identifying bipartite components.  Current confluent graph drawing algorithms have tended to focus on the computability of planar confluent graphs.  There has been little focus on developing methods that can do something reasonable when a planar drawing is not possible.  We think an optimisation-based approach related to the techniques described in this paper may have some success in this regard.

Finally, we need better layout methods for power-graph decompositions and clustered graphs generally.  The examples in this paper were initially arranged with the Y-Files organic layout method but required significant manipulation to achieve pleasing alignment and to minimise crossings between edges.

\bibliographystyle{abbrv}
\bibliography{PowerGraphAnalysis}

\appendix
\clearpage
\noindent \textbf{\textsc{\textsf{\Large Appendix}}}
\vspace{2mm}

\noindent For the published version of this paper these additional details will be 
provided as an on-line tech report.
\section{ILP model}
The following is a complete description of the ILP model used to
minimise the number of power edges (or equivalently to maximise the
number of edges saved by adding modules).
\noindent\underline{\textbf{Input and Parameters}}\vspace{1mm}

\noindent $n$  is the number of vertices of the input graph $G$.\\
$V = \{0,1,\dots,n-1\}$ represents the vertices of the input
  graph $G$.\\
$e(u,v)$  represents the edges of the input graph $G$ as an
  incidence matrix. That is, $e(u,v) = 1$ if $(u,v)$ is an edge of $G$
  and $e(u,v)=0$ otherwise.\\
$m$ is the number of modules with at least two elements (we
  consider each singleton vertex to belong to its own module).\\
$M= \{ 0,1,\dots,n+m-1\}$ represents the set of all modules. 

\noindent \underline{\textbf{Integer decision variables}}\vspace{1mm}

\noindent $sav[m_1,m_2]$ the number of edges that may be
removed from $G$ (to then be replaced by a single edge) if the modules
$m_!$ and $m_2$ are added.

\noindent \underline{\textbf{Binary decision variables}}\vspace{1mm}

\noindent $mod[v,m]$ takes the value $1$ if and only if vertex $v$
  belongs to the module $m$.\\
$ind[v,m]$ takes the value $1$ if and only if $(v,u)$ is an
  edge, for all $u$ in the module $m$.\\
$bic[m_1,m_2]$ takes the value $1$ if and only if, for
  every vertex $v\in m_1$ and every vertex $u\in m_2$, the pair
  $(u,v)$ is an edge in $G$.\\
$dis[m_1,m_2]$ takes the value $1$ if and only if the modules $m_1$
  and $m_2$ are disjoint sets of vertices.\\
$sub[m_1,m_2]$ takes the value $1$ if and only if the module
  $m_1$ is a proper subset of the module $m_2$.\\
$mInd[v,m_1,m_2]$ takes value $1$ if and only if $v\in m_1$ and $ind[v,m_2]=1$.\\
$vMod[v,m_1,m_2]$ takes value $1$ if and only if $v\in m_1$ and $v\in m_2$.\\
$sVer[v_1,v_2,m_1,m_2]$ takes value $1$ if and only if
  $(v_1,v_2)$ is an edge with $v_1\in m_1$ and $v_2\in m_2$ and the
  edge $(v_1,v_2)$ can be removed if $m_1$ and $m_2$ are added.\\
$sMod[m_1,m_2]$ takes value $1$ if and only if $sav[m_1,m_2] > 0$.

\noindent \underline{\textbf{Objective}}\vspace{1mm}

\noindent Maximise \[\sum\{sav[m_1,m_2] - sMod[m_1,m_2] \mid m1,m2\in M,\ m_1\neq m_2\}.\]

\noindent \underline{\textbf{Constraints}}\vspace{1mm}
\begin{enumerate}
\itemsep0em 
\item $\sum\limits_{u\in V}(e(v,u)-1)mod[u,m_1] \geq n(ind[v,m]-1)$,  $v\in V$, $m_1\in M$.
\item $\sum\limits_{u\in V}(e(v,u)-1)mod[u,m_1] \leq ind[v,m_1]-1$,  $v\in V$, $m_1\in M$.
\item $mInd[v,m_1,m_2] \leq mod[v,m_1]$,  $v\in V$, $m_1\neq m_2\in M$. 
\item $mInd[v,m_1,m_2] \leq ind[v,m_1]$,  $v\in V$, $m_1\neq m_2\in M$.
\item $mInd[v,m_1,m_2] \geq mod[v,m_1] + ind[v,m_1] -1$,  $v\in V$, $m_1\neq m_2\in M$. 
\item $\sum\limits_{v\in V} (mInd[v,m_1,m_2]-mod[v,m_1]) \geq n(bic[m_1,m_2]-1)$, $m_1\neq m_2\in M$.
\item $\sum\limits_{v\in V} (mInd[v,m_1,m_2]-mod[v,m_1]) \leq bic[m_1,m_2]-1$,  $m_1\neq m_2\in M$.
\item $vMod[v,m_1,m_2] \leq mod[v,m_1]$,  $v\in V$, $m_1\neq m_2\in M$. 
\item $vMod[v,m_1,m_2] \leq mod[v,m_2]$,  $v\in V$, $m_1\neq m_2\in M$. 
\item $vMod[v,m_1,m_2] \geq mod[v,m_1] + mod[v,m_2] -1$,  $v\in V$, $m_1\neq m_2\in M$. 
\item $\sum\limits_{v\in V}vMod[v,m_1,m_2] \leq n(1-dis[m_1,m_2])$, $m_1\neq m_2\in M$.
\item $\sum\limits_{v\in V}vMod[v,m_1,m_2] \geq 1-dis[m_1,m_2]$, $m_1\neq m_2\in M$.
\item $dis[m_1,m_2] = dis[m_2,m_1]$, $m_1\neq m_2\in M$.
\item $\sum\limits_{v\in V}(vMod[v,m_1,m_2] - mod[v,m_1])\geq n(sub[m_1,m_2]-1)$,  $m_1\neq m_2\in M$.
\item $\sum\limits_{v\in V}(vMod[v,m_1,m_2] - mod[v,m_1])\leq sub[m_1,m_2]-1$,  $m_1\neq m_2\in M$.
\item $\sum\limits_{v\in V}vMod[v,m_1,m_2]\leq (\sum\limits_{v\in V}mod[v,m_2]) - sub[m_1,m_2]$,  $m_1\neq m_2\in M$.
\item $dis[m_1,m_2] + sub[m_1,m_2] + sub[m_2,m_1] = 1$, $m_1\neq m_2\in M$.
\item $sVer[v_1,v_2,m_1,m_2] \leq e(v_1,v_2)$,  $v_1,v_2\in V$, $m_1,m_2\in M$.
\item $sVer[v_1,v_2,m_1,m_2] \leq mod[v_1,m_1]$,  $v_1,v_2\in V$, $m_1,m_2\in M$.
\item $sVer[v_1,v_2,m_1,m_2] \leq mod[v_2,m_2]$,  $v_1,v_2\in V$, $m_1,m_2\in M$.
\item $sVer[v_1,v_2,m_1,m_2] \leq bic[m_1,m_2]$,  $v_1,v_2\in V$, $m_1,m_2\in M$.
\item $\sum\{sVer[v_1,v_2,m_1,m_2] \mid m_1,m_2\in M, m_1\neq m_2 \}\leq 1$,  $v_1,v_2\in V$.
\item $sav[m_1,m_2]\leq \sum\{sVer[v_1,v_2,m_1,m_2] \mid v_1,v_2\in V, v_1\neq v_2 \}$,  $m_1\neq m_2\in M$.
\item $sMod[m_1,m_2] \leq sav[m_1,m_2]$,  $m_1\neq m_2\in M$.
\item $sav[m_1,m_2] \leq n^2 sMod[m_1,m_2]$, $m_1\neq m_2\in M$.
\item $\sum\limits_{v\in V} mod[v,m_1] =1$,  $m_1\in V$.
\item $mod[m_1,m_1] = 1$,  $m_1\in V$.
\end{enumerate}

Constraints 1 and 2 define the variables
$ind[v,m]$. Constraints 3 -- 5 define
$mInd[v,m_1,m_2]$. Constraints 6 and 7 define
$bic[m_1,m_2]$. Constraints 8 -- 10 define
$vMod[v,m_1,m_2]$. Constraints 11 -- 13 define
$dis[m_1,m_2]$. Constraints 14 --16 define
$sub[m_1,m_2]$. Constraint 17 says that for any two distinct
modules $m_1$ and $m_2$, either $m_1$ and $m_2$ are disjoint, or one
is a proper subset of the other. Constraint 18 -- 21 defines the
variables $sVer[v_1,v_2,m_1,m_2]$. Constraint 22 says that no
edge can be counted twice in the saving calculation. Constraint 23
defines the variables $sav[m_1,m_2]$. Constraints 24 and 25
defines $sMod[m_1,m_2]$. Constraints 26 and 27 force each
vertex to be a singleton module.

\end{document}